\documentclass[conference]{IEEEtran}
\IEEEoverridecommandlockouts
\usepackage{amsfonts}
\usepackage{enumerate,fancybox,url,ascmac,here,comment,wrapfig,cite}
\usepackage{amsmath,amssymb,amsfonts,bm}
\usepackage{algorithmic,algorithm}
\usepackage{graphicx}
\usepackage{textcomp}
\usepackage{mathtools}
\usepackage{xcolor}
\usepackage{ascmac}
\usepackage{physics}
\usepackage{tikz}
\usepackage{mathrsfs}
\usepackage[inline]{enumitem}

\newtheorem{theo}{Theorem}
\newtheorem{lem}{Lemma}
\newtheorem{coro}{Corollary}

\newtheorem{exam}{Example}

\def\BibTeX{{\rm B\kern-.05em{\sc i\kern-.025em b}\kern-.08em
    T\kern-.1667em\lower.7ex\hbox{E}\kern-.125emX}}

\newcommand{\mat}[1]{\mathbf{#1}}

\newcommand{\state}[2]{S(\C^{#1 \otimes #2})}
\newcommand{\QU}[1]{\mathcal{#1}}

\newcommand{\Z}{\mathbb{Z}}
\newcommand{\PZ}{\mathbb{Z}^+}
\newcommand{\C}{\mathbb{C}}
\newcommand{\F}{\mathbb{F}}
\newcommand{\setint}[1]{[\hspace{-0.5mm}[#1]\hspace{-0.5mm}]}
\newcommand{\subint}[2]{\binom{\setint{#1}}{#2}}

\begin{document}

\title{Decoding Algorithm to Composite Errors Consisting of Deletions and Insertions\\ for Quantum Deletion-Correcting Codes\\ Based on Quantum Reed-Solomon Codes}

\author{
  \IEEEauthorblockN{Koki Sasaki, Ken Nakamura, and Takayuki Nozaki}
  \IEEEauthorblockA{
    Dept.\ of Informatics, Yamaguchi University, JAPAN \\
    Email: {\tt \{f006wbw, tnozaki\}@yamaguchi-u.ac.jp}
  }
}

\maketitle

\begin{abstract}
  This paper focuses on Hagiwara codes,
  which are quantum deletion-correcting codes constructed by 
  the quantum Reed-Solomon codes.
  Although Hagiwara codes
  can correct
  composite errors consisting of deletions and insertions,
  an efficient decoding algorithm to such errors
  remains
  an open problem. 
  In this paper,
  we provide a decoding algorithm to such errors
  for Hagiwara codes.
\end{abstract}

\section{Introduction}
In quantum information theory,
information is
represented by the states of quantum systems.
A quantum insertion error
and
a quantum deletion error \cite{leahy2019quantum}
change
the information.
Here,
a quantum deletion error
loses
a qubit of a multi-qubit system
and
a quantum insertion error
adds
a new qubit to a multi-qubit system.

Nakayama and Hagiwara constructed
the first quantum single-deletion-correcting code \cite{nakayama2020first}
in $2020$.
Following this study,
quantum deletion codes
\cite{hagiwara2020four,matsumoto2022constructions,hagiwara2023quantum,hagiwara2025quantum}
have been constructed
and
its decoding algorithms
\cite{hagiwara2020four,matsumoto2022constructions,hagiwara2023quantum,hagiwara2025quantum,hagiwara2021four,nakamura2025multiple,DBLP:conf/isita/SasakiN24,nakamura2024decoding,it2025-3}
have been proposed.
Although these codes can correct
composite errors consisting of deletions and insertions \cite{nakamura2026insertion},
most algorithms
correct to only deletions or only insertions
and
there are a few studies \cite{nakamura2024decoding,it2025-3}
to composite errors
consisting of deletions and insertions.

This paper focuses on
{\it Hagiwara codes} \cite{hagiwara2023quantum,hagiwara2025quantum},
which are quantum $t$-deletion-correcting codes constructed by 
inserting the {\it marker} qubits
to the quantum Reed-Solomon (RS) codes \cite{grassl1999quantum}.
The marker helps to estimate the range of position of deleted qubits
and
transforms to a recoverable state by the quantum RS code.
For these codes,
we proposed
a decoding algorithm to
quantum $t$-insertion errors \cite{DBLP:conf/isita/SasakiN24}
and
that to quantum single-insertion and single-deletion error \cite{it2025-3}.

This paper propose
a decoding algorithm
to classical $t_d$-deletion errors and classical $t_i$-insertion errors
for {\it classical Hagiwara codes},
which
is the classical counterpart of Hagiwara code, 
by generalizing \cite{it2025-3},
where $t_d+ t_i \leq t$ and $t_d,t_i \neq 0$.
Then,
this paper extends this classical algorithm to quantum.

\section{Preliminary}
\subsection{Notations}
Let $\PZ$ and $\C$ be 
the sets of all positive integers and complex numbers, respectively.
Let $\F_q$ be the finite field of order $q$.
We denote
the $n$-dimensional vector space over the field $\F_q$ (resp.~$\C$)
by $\F_q^n$ (resp.~$\C^n$).
Let 
$\mat{M}^{\dagger}$ be
the adjoint of a matrix $\mat{M}$
and
let
$\mat{M}_1 \otimes \mat{M}_2$ be
the Kronecker product of matrices $\mat{M}_1$ and $\mat{M}_2$.
For $n,m,t \in \PZ$,
define
$\setint{m,n} := \{ x \in \PZ \mid m \leq x \leq n \}$,
$\setint{n} := \setint{1,n}$,
$\subint{n}{t} := \{ A \subset \setint{n} \mid |A| = t \}$,
and
$\subint{n}{\leq t} := \{ A \subset \setint{n} \mid |A| \leq t \}$.

To simplify the notation,
we write a sequence $\bm{x} = (x_1, x_2, ..., x_n) \in \F_q^n$
as $x_1 x_2 ... x_n$.
We denote the subsequence $( x_{i_1}, x_{i_2}, ..., x_{i_m} )$ of $\bm{x} = (x_1, x_2, ..., x_n)$
by $\bm{x}_{\{ i_1, i_2, ..., i_m \}}$.

Define
$\ket{0} := \binom{1}{0}$
and
$\ket{1} := \binom{0}{1}$.
For $\bm{x} = (x_1, x_2, ..., x_n) \in \F_2^n$,
define
$\ket{\bm{x}} := \ket{x_1} \otimes \ket{x_2} \otimes \cdots \otimes \ket{x_n} \in \C^{2^n}$.
In this paper,
we denote $\C^{2^n}$
by $\C^{2 \otimes n}$.
Given a complex column vector $\ket{\psi}$,
we denote $\ket{\psi}^\dagger$
by $\bra{\psi}$.

An $n$ qubit system can be written as
a $2^n$-dimensional density matrices
$\rho = \sum_{ \bm{x},\bm{y} \in \F_2^n } m_{\bm{x},\bm{y}} \ket{\bm{x}}\bra{\bm{y}}$,
where
$m_{\bm{x},\bm{y}} \in \C$
and
$\sum_{ \bm{x},\bm{y} \in \F_2^n } m_{\bm{x},\bm{y}} = 1$.
Here,
a density matrix
is
a positive-semidefinite matrix
such that
the trace equals $1$.
Let $\state{2}{n}$ be
the set of all $2^n$-dimensional density matrix.
A quantum state $\rho$ is pure,
if $\rho$ has rank $1$.
Otherwise $\rho$ is mixed.
When $\rho \in \state{2}{n}$ is pure,
there exists a vector $\ket{\phi} \in \C^{2 \otimes n}$
such that $\rho = \ket{\phi}\bra{\phi}$.
Then,
we denote $\rho$, by $\ket{\phi}$.

\subsection{Classical Substitution/Erasure/Deletion/Insertion Error}
A classical substitution error
changes
a symbol of a sequence
to another.
In particular,
a classical substitution error is called
a classical erasure error
if
it changes to the erasure symbol $?$.
For example,
the erasure error at $\{ 3 \}$
changes $01110$ to $01?10$.

A classical deletion error
loses
a symbol of a sequence.
For $\bm{x} \in \F_q^n$,
define classical $t$-deletion errors $D_J$ at $J \in \subint{n}{t}$ as
$D_J(\bm{x}) := \bm{x}_{\setint{n} \setminus J}$.

A classical insertion error
adds
a new symbol to a sequence.
For $\bm{x} \in \F_q^n$,
define classical $t$-insertion errors $I_K$ at $K \in \subint{n+t}{t}$ as
$I_K(\bm{x}) := \{ \bm{y} \in \F_q^{n+t} \mid D_K(\bm{y}) = \bm{x} \}$.

\subsection{Quantum Deletion/Insertion/Erasure Error}
A quantum deletion error
loses
a qubit of a multi-qubit system.
For $i \in \setint{n}$ and $\mat{M} = \sum_{ \bm{x},\bm{y} \in \F_2^n } m_{\bm{x},\bm{y}} \ket{\bm{x}}\bra{\bm{y}}$,
define the partial trace $\text{tr}_i$ as 
$\text{tr}_i(\mat{M}) = \sum_{ \bm{x},\bm{y} \in \F_2^n } m_{\bm{x},\bm{y}} \text{tr}(\ket{x_i}\bra{y_i}) \ket{D_{\{i\}}(\bm{x})}\bra{D_{\{i\}}(\bm{y})}$,
where
$\text{tr}(\cdot)$ is the trace of a square matrix.
For $\rho \in \state{2}{n}$,
define quantum $t$-deletion errors $\QU{D}_J$
at $J = \{ j_1, j_2, ..., j_t \} \in \subint{n}{t}$ with $j_1 < j_2 < \cdots < j_t$ as
$\QU{D}_J(\rho) := \text{Tr}_{j_1} \circ \text{Tr}_{j_2} \circ \cdots \text{Tr}_{j_t}( \rho )$.

A quantum insertion error
adds
a new qubit to a multi-qubit system.
For $\rho \in \state{2}{n}$,
define quantum $t$-insertion errors $\QU{I}_K$ at $K \in \subint{n+t}{t}$ as
$\QU{I}_K(\rho) := \{ \rho' \in \state{2}{(n+t)} \mid \QU{D}_K(\rho') = \rho \}$.

\section{Hagiwara Codes \cite{hagiwara2023quantum,hagiwara2025quantum}}
\label{sec:hagi_code}
      {\it Hagiwara code} \cite{hagiwara2023quantum,hagiwara2025quantum}
is
a multiple deletion-correcting quantum code,
which is constructed
by inserting the {\it marker} qubits
to the quantum RS codes \cite{grassl1999quantum}.
The marker helps to estimate the range of position of deleted qubits
and
transforms quantum $t$-deletion errors to quantum $t$-erasure errors 
\cite{hagiwara2023quantum,hagiwara2025quantum}.
Section~\ref{subsec:C_del_code} gives
the classical counterpart of Hagiwara code.
Section~\ref{subsec:def_hagi_code}
introduces the quantum RS codes and Hagiwara codes.

\subsection{Classical Deletion-Correcting Codes by the Marker}
\label{subsec:C_del_code}
\subsubsection{Binary Expansion}
Let $B = \{ b_1, .., b_E \} \subset \F_{p^E}$
be a basis of $\F_{p^E}$ over $\F_p$,
i.e.,
$\F_{p^E} = \{ \sum_{i=1}^E a_i b_i \mid a_i \in \F_p \}$.
For such a basis $B$,
the dual basis $B^{\perp}$ of $B$
is defined as
$B^{\perp} := \{ b'_1, .., b'_E \} \subset \F_{p^E}$
such that
$\text{Tr}( b_i b_j' ) = \delta_{i,j}$ ($i,j = 1,2,...,E$),
where
$\text{Tr}(\beta) := \sum_{i=0}^{E-1} \beta^{p^i}$
and
$\delta_{i,j}$ is the Kronecker delta.

A basis $B$ is self-dual,
if $B^{\perp} = B$.
By using a self-dual basis $B = \{ b_1, b_2, ..., b_E \}$ of $\F_{2^E}$ over $\F_2$,
any element $\beta = \sum_{i=1}^E a_i b_i \in \F_{2^E}$ can be expanded as
a binary sequence $( a_1, a_2, ..., a_E ) \in \F_2^{E}$.
Then,
we denote this binary sequence by $B(\beta) := ( a_1, a_2, ..., a_E )$.
Define the binary expansion $B(C)$ of a linear code $C \subset \F_{2^E}^N$ as follows:
\begin{align*}
  B(C)
  &:=
  \{
  \big(
  \bm{c}_1,
  ...,
  \bm{c}_N
  \big)
  \in
  \F_2^{EN}
  \mid
  \bm{c}_i = B(c'_i),
  \bm{c}' = (c'_i) \in C
  \}.
\end{align*}

\subsubsection{Classical Counterpart of Hagiwara Codes}
The sequence $\bm{0}^{(t)}$ (resp.~$\bm{1}^{(t)}$) represents
the all-zero (resp.~all-one) sequence of length $t$.
The concatenation $\bm{0}^{(t)}\bm{1}^{(t)}$ of $\bm{0}^{(t)}$ and $\bm{1}^{(t)}$
is called the {\it marker}.
For a classical $t$-erasure-correcting code $C$,
the codeword of classical Hagiwara code
corresponding to $( \bm{c}_1, \bm{c}_2, ..., \bm{c}_N ) \in B(C)$
is
\begin{align}
  \bm{c}
  =
  \bm{c}_1 \bm{0}^{(t)}\bm{1}^{(t)}
  \bm{c}_2 \bm{0}^{(t)}\bm{1}^{(t)}
  \cdots
  \bm{c}_N \bm{0}^{(t)}\bm{1}^{(t)}.
  \label{eq:C_hagi_codeword}
\end{align}
In this paper,
we call $\bm{c}_b$ the {\it $b$th block}.
This code corrects $t$-deletion error
by a similar approach to Hagiwara codes.
This code is also
classical $t_d$-deletion and $t_i$-insertion-correcting code
since
classical $t$-deletion-correcting codes
also correct $t_d$-deletions and $t_i$-insertions \cite{levenshtein1966binary},
where $t_d + t_i \leq t$.

\subsubsection{Classical Block-Substitution/Erasure}
We call
a transformation to
a sequence $\bm{c}'_b \in \F_2^E$ from $\bm{c}_b$
a {\it classical block-substitution} at $b$,
e.g.,
$\bm{c}_b = 0110 \mapsto \bm{c}'_b = 1010$.
In particular,
we call a classical block-substitution at $b$
a {\it classical block-erasure error} at $b$
if
the receiver knows the error position $b$.

\subsection{Definition of Hagiwara Codes \cite{hagiwara2023quantum,hagiwara2025quantum}}
\label{subsec:def_hagi_code}
Let $C_1$ be an $(N,K_1)$-RS code over $\F_{2^E}$
and 
$C_2$ be an $(N,K_2)$-RS code over $\F_{2^E}$
such that $C_2 \subset C_1$.
Then,
the Quantum RS code $\QU{R}$ \cite{grassl1999quantum}
is defined
by the Calderbank–Shor–Steane (CSS) code \cite{calderbank1996good,steane1996multiple} of
$B(C_1)$ over $B(C_2)$,
i.e.,
\begin{align*}
  \QU{R}
  &:=
  \left\{
  \textstyle
  \sum_{i=1}^k
  \alpha_i \ket{\psi^{(i)}}
  \middle|
  \alpha_i \in \C,
  \textstyle
  \sum_{i=1}^k
  |\alpha_i|^2
  =
  1
  \right\}
  \subset
  \C^{2 \otimes NE},
\end{align*}
where
$k := 2^{K_1-K_2}$ and
$\ket{\psi^{(i)}} := \sum_{\bm{x} \in X^{(i)}} \ket{\bm{x}} / \sqrt{ |X^{(i)}| }$
for $X^{(i)} \in B(C_1) / B(C_2)$.
For $b \in \setint{N}$,
let $R_b$ be
the system made up of
the $(b-1)E+1$, $(b-1)E+2$, ..., $(b-1)E+E$th qubits of $\rho \in \QU{R}$.
In this paper,
we call $R_b$ the {\it $b$th block}.

For simplicity,
we write $\ket{\phi_1} \ket{\phi_2}$
for the Kronecker product $\ket{\phi_1} \otimes \ket{\phi_2}$.
Hagiwara codes
is constructed by inserting the marker $\ket{ \bm{0}^{(t)} \bm{1}^{(t)} }$
to quantum RS codes,
similar to classical Hagiwara codes.
Mathematically,
for a quantum RS code $\QU{R}$ and
a codeword $\ket{\psi} = \sum_{ \bm{x} = (\bm{x}_1,\bm{x}_2,...,\bm{x}_N) \in \F_2^{NE} } \alpha_{\bm{x}} \ket{\bm{x}_1} \ket{\bm{x}_2} \cdots \ket{\bm{x}_N} \in \QU{R}$
($\bm{x}_i \in \F_2^{E}$),
the corresponding codeword in Hagiwara code is given by
\begin{align*}
  \ket{\phi}
  =
  \hspace*{-3mm}
  \sum_{ \bm{x} \in \F_2^{NE} }
  \hspace*{-3mm}
  \alpha_{\bm{x}}
  \ket{\bm{x}_1} \ket{ \bm{0}^{(t)} \bm{1}^{(t)} }
  \ket{\bm{x}_2} \ket{ \bm{0}^{(t)} \bm{1}^{(t)} }
  \cdots
  \ket{\bm{x}_N} \ket{ \bm{0}^{(t)} \bm{1}^{(t)} }.
\end{align*}

\section{Decoding Algorithm to Composite Errors for Hagiwara Codes}
In this section,
we present
a decoding algorithm for classical Hagiwara codes
to composite errors consisting of deletions and insertions.
In addition,
we extend this classical algorithm to quantum.

\subsection{Decoding Algorithm for Classical Hagiwara Codes}
\label{subsec:deco_C_hagi_code}
\subsubsection{Notations}
Let $\bm{c} \in \F_2^{N(E+2t)}$
be a codeword defined in Eq.~\eqref{eq:C_hagi_codeword}
and
let $\bm{y} \in \F_2^{N(E+2t)+r}$
be a received word
given by $t_d$-deletions and $t_i$-insertions to $\bm{c}$,
where
$t_d, t_i \ge 1$
and
$t_d + t_i \leq t$.
Then,
an upper bound on $t_d$ and $t_i$
is given by
$t_d \leq \lfloor (t-r)/2 \rfloor =: \tau_d$ and
$t_i \leq \lfloor (t+r)/2 \rfloor =: \tau_i$.
Define $\tau := \tau_d + \tau_i$.

For $b \in \PZ$,
we denote
the position of the previous symbol of $\bm{c}_b$
by $\beta_b$
and
the position of the last $0$ in the $b$th marker
by $\gamma_b$,
i.e.,
$\beta_b := (E+2t)(b-1)$
and
$\gamma_b := \beta_b + E + t$.
Define
$\bm{y}_b := \bm{y}_{\setint{\gamma_b-\tau_d+1,\gamma_b+\tau_i}}$.
Figure~\ref{fig: x_and_y} illustrates
the position $\beta_b$ and $\gamma_b$
and
range of $\bm{y}_b$.
Let $\bm{?}^{(t)}$ be
the $t$-repetition of the erasure symbol $?$.

\begin{figure}[t]
  \centering
      {\small
        \begin{tikzpicture}
          
          \node at (-2.0,1.8) {$\bm{c}$};

          \node at (-1.0,1.8) {$\cdots$};

          \draw[dotted, fill=gray, fill opacity=0.5] (-0.2,1.65) rectangle (0,2.0);
          \node at (-0.1,1.4) {$\beta_b$};
          \draw (0,1.65) rectangle (1.4,2.0);
          \node at (0.7,1.8) {$\bm{c}_b$};
          \draw (1.4,1.65) rectangle (2.2,2.0);
          \node at (1.8,1.8) {$0 \cdots 0$};
          \draw[dotted, fill=gray, fill opacity=0.5] (2.8,1.65) rectangle (3.0,2.0);
          \node at (2.9,1.4) {$\gamma_b$};
          \draw (2.2,1.65) rectangle (3.0,2.0);
          \node at (2.6,1.82) {$\bm{0}^{(\tau_d)}$};
          \draw (3.0,1.65) rectangle (3.8,2.0);
          \node at (3.4,1.82) {$\bm{1}^{(\tau_i)}$};
          \draw (3.8,1.65) rectangle (4.6,2.0);
          \node at (4.2,1.8) {$1 \cdots 1$};
          
          \node at (5.2,1.8) {$\cdots$};

          \begin{scope}[shift={(-0.4,-1.45)}]
            \node at (-1.6,1.8) {$\bm{y}$};
            
            \node at (-0.8,1.8) {$\cdots$};
            
            \draw[dotted, fill=gray, fill opacity=0.5] (0.2,1.65) rectangle (0.4,2.0);
            \draw (0,1.65) rectangle (1.4,2.0);
            \node at (0.7,1.8) {$\bm{c}_b$};
            \draw (1.4,1.65) rectangle (2.2,2.0);
            \node at (1.8,1.8) {$0 \cdots 0$};
            \draw[dotted, fill=gray, fill opacity=0.5] (3.2,1.65) rectangle (3.4,2.0);
            \draw (2.2,1.65) rectangle (3.0,2.0);
            \node at (2.6,1.82) {$\bm{0}^{(\tau_d)}$};
            \draw (3.0,1.65) rectangle (3.8,2.0);
            \node at (3.4,1.82) {$\bm{1}^{(\tau_i)}$};
            \draw (3.8,1.65) rectangle (4.6,2.0);
            \node at (4.2,1.8) {$1 \cdots 1$};
            
            \node at (5.2,1.8) {$\cdots$};
          \end{scope}

          \draw[dashed] (1.4,1.65) -- (1.0,0.55);
          \draw[dashed] (4.6,1.65) -- (4.2,0.55);
          
          \draw[dashed] (2.2,2.0) -- (2.2,0.55);
          \draw[dashed] (3.8,2.0) -- (3.8,0);
          \fill[fill=gray, fill opacity=0.3] (2.2,2.0) rectangle (3.8,0.2);
          
          \draw[<->, >=stealth] (2.2,0.05) -- (3.8,0.05);
          \node at (3.0,-0.2) {$\bm{y}_b$};

          \draw[<->, >=stealth] (2.2,2.2) -- (3.0,2.2);
          \node at (2.6,2.4) {$\tau_d$};

          \draw[<->, >=stealth] (3.0,2.2) -- (3.8,2.2);
          \node at (3.4,2.4) {$\tau_i$};

        \end{tikzpicture}
      }
      \caption{$\beta_b$, $\gamma_b$, and $\bm{y}_b$}
      \label{fig: x_and_y}
\end{figure}

\subsubsection{Algorithm}
Algorithm~\ref{algo:tran_era}
transforms from $\bm{y}$
to a sequence $\bm{z}$,
which is sequence $(c_1, c_2, ..., c_N) \in B(C)$
with $e$ block-erasures at $P$ and $m$ block-substitutions ($e + 2m \leq t$),
shown in Section~\ref{subsec:just}.
Then,
we get
the message from $\bm{z}$
by the error and erasure correcting algorithm for the RS code.

We describe the principle of Algorithm~\ref{algo:tran_era}.
In the case of insertion error,
there are two types of events;
{\it marker-preserve}
and
{\it marker-destruction}.
In Marker-preserve,
the format of the marker is preserved $\bm{t}_b = \bm{0}^i \bm{1}^{\tau-i}$.
This event occurs
when
symbols are inserted in $\bm{c}_b$,
$0$ is inserted in zeros of the $b$th marker,
or
$1$ is inserted in ones of the $b$th marker.
On the other hand,
in marker-destruction,
the format of the marker is not preserved.
This event occurs
when
$1$ is inserted in zeros of the $b$th marker
or
$0$ is inserted in ones of the $b$th marker.
Note that,
in the case of deletion errors,
the format of the marker is preserved.
If marker-preserve occurs,
by the $b$th marker
and
the number $u$ (resp.~$v$) of detected deletion (resp.~insertion) errors,
this algorithm
detects the number of deletions/insertions in $\bm{c}_b$
and
performs
as follows:
\begin{enumerate*}[label={(\roman*)}]
\item if $w$-deletions (Line $7$) or some-insertions (Line $9$)
  are detected,
  the algorithm
  outputs $\bm{z}_b = \bm{?}^{(E)}$ and
  adds the position $b$ to the set $P$ of the erased blocks,
\item if no deletions or insertions (Line $5$) are detected,
  it outputs $\bm{z}_b = \bm{y}_{ \setint{ \beta_b+l+1, \beta_b+l+E } }$.
\end{enumerate*}
On the other hand,
if marker-destruction occurs,
since some-insertions occurs,
the algorithm
detects single-insertion
and
outputs $\bm{z}_b = \bm{?}^{(E)}$
in Line $10$.

\begin{figure}[t]
  \begin{algorithm}[H]
    \caption{Transformation to Block-Substitutions and Erasures}
    \label{algo:tran_era}
    \begin{algorithmic}[1] 
      \REQUIRE $\bm{y} \in \F_2^{N(E+2t)+r}$.
      \ENSURE $\bm{z} \in \{ 0,1,? \}^{NE}$, $P \in \subint{N}{\leq t}$.
      \STATE $\tau_d \gets \lfloor (t-r)/2 \rfloor$,
      $\tau_i \gets \lfloor (t+r)/2 \rfloor$,
      $\tau \gets \tau_d + \tau_i$.
      \STATE $(u,v) \gets (0,0)$,
      $P \gets \emptyset$.
      \FOR{$b = 1,2,..., N$}
      \STATE $l \gets -u+v$.
      \IF{$\bm{y}_b = \bm{0}^{(\tau_d+l)} \bm{1}^{(\tau_i-l)}$}
      \STATE $\beta_b \gets (E+2t)(b-1)$,
      $\bm{z}_b \gets \bm{y}_{ \setint{ \beta_b+l+1, \beta_b+l+E } }$.
      \ELSIF{$\bm{y}_b = \bm{0}^{(\tau_d+l-w)} \bm{1}^{(\tau_i-l+w)}$ for some $w \in \PZ$}
      \STATE $\bm{z}_b \gets \bm{?}^{(E)}$,
      $P \gets P \cup \{ b \}$,
      $u \gets u+w$.
      \ELSE
      \STATE $\bm{z}_b \gets \bm{?}^{(E)}$,
      $P \gets P \cup \{ b \}$,
      $v \gets v+1$.
      \ENDIF
      \ENDFOR
      \STATE $\bm{z} \gets (\bm{z}_1, \bm{z}_2, ..., \bm{z}_N)$.
    \end{algorithmic}
  \end{algorithm}
\end{figure}

\subsubsection{Example}
For a classical $t=4$-erasure-correcting code $C \subset \F_{2^E}^5$
and
the codeword $\bm{c} = \bm{c}_1\bm{0}^{(4)}\bm{1}^{(4)} \cdots \bm{c}_5\bm{0}^{(4)}\bm{1}^{(4)}$
in classical Hagiwara code based on $C$,
assume that
the received word $\bm{y}$
is given by
$1$-insertion to $\bm{c}_1$,
$1$-deletion to $\bm{c}_3$,
and
$1$-deletion and $1$-insertion to $\bm{c}_4$,
as shown in Fig.~\ref{fig: ex}.
If we input $\bm{y}$,
Algorithm~\ref{algo:tran_era} performs
as follows.
Note that
$\tau_d = \tau_i = 2$
since $r = 0$.

\begin{figure*}[t]
  \centering
  \begin{tikzpicture}  
    \node at (-0.7,1.8) {$\bm{c}$};

    \draw (0,1.65) rectangle (1.0,2.0);
    \node at (0.5,1.8) {$\bm{c}_1$};
    \draw (0.05,2.1) -- (0.05,2.2) -- (0.95,2.2) -- (0.95,2.1);
    \node at (0.5,2.35) {\footnotesize $E$};
    \draw (1.05,2.1) -- (1.05,2.2) -- (1.75,2.2) -- (1.75,2.1);
    \node at (1.4,2.35) {\footnotesize $t$};
    \draw (1.0,1.65) rectangle (1.8,2.0);
    \foreach \x in {0,1,2,3} {\node at (1.1+0.2*\x,1.8) {$0$};}
    \draw (1.85,2.1) -- (1.85,2.2) -- (2.55,2.2) -- (2.55,2.1);
    \node at (2.2,2.35) {\footnotesize $t$};
    \draw (1.8,1.65) rectangle (2.6,2.0);
    \foreach \x in {0,1,2,3} {\node at (1.9+0.2*\x,1.8) {$1$};}
    \draw (0.45,1.4) -- (0.5,1.55) -- (0.55,1.4);
    \node at (0.5,1.25) {$1$};
    
    \begin{scope}[shift={(2.6,0)}]
      \draw (0,1.65) rectangle (1.0,2.0);
      \node at (0.5,1.8) {$\bm{c}_2$};
      \draw (1.0,1.65) rectangle (1.8,2.0);
      \foreach \x in {0,1,2,3} {\node at (1.1+0.2*\x,1.8) {$0$};}
      \draw (1.8,1.65) rectangle (2.6,2.0);
      \foreach \x in {0,1,2,3} {\node at (1.9+0.2*\x,1.8) {$1$};}
    \end{scope}
    
    \begin{scope}[shift={(5.2,0)}]
      \draw (0,1.65) rectangle (1.0,2.0);
      \node at (0.5,1.8) {$\bm{c}_3$};
      \draw (1.0,1.65) rectangle (1.8,2.0);
      \foreach \x in {0,1,2,3} {\node at (1.1+0.2*\x,1.8) {$0$};}
      \draw (1.8,1.65) rectangle (2.6,2.0);
      \foreach \x in {0,1,2,3} {\node at (1.9+0.2*\x,1.8) {$1$};}
      \node at (0.5,1.5) {$\times$};
    \end{scope}
    
    \begin{scope}[shift={(7.8,0)}]
      \draw (0,1.65) rectangle (1.0,2.0);
      \node at (0.5,1.8) {$\bm{c}_4$};
      \draw (1.0,1.65) rectangle (1.8,2.0);
      \foreach \x in {0,1,2,3} {\node at (1.1+0.2*\x,1.8) {$0$};}
      \draw (1.8,1.65) rectangle (2.6,2.0);
      \foreach \x in {0,1,2,3} {\node at (1.9+0.2*\x,1.8) {$1$};}
      \draw (0.35,1.4) -- (0.4,1.55) -- (0.45,1.4);
      \node at (0.4,1.25) {$0$};
      \node at (0.6,1.5) {$\times$};
    \end{scope}
    
    \begin{scope}[shift={(10.4,0)}]
      \draw (0,1.65) rectangle (1.0,2.0);
      \node at (0.5,1.8) {$\bm{c}_5$};
      \draw (1.0,1.65) rectangle (1.8,2.0);
      \foreach \x in {0,1,2,3} {\node at (1.1+0.2*\x,1.8) {$0$};}
      \draw (1.8,1.65) rectangle (2.6,2.0);
      \foreach \x in {0,1,2,3} {\node at (1.9+0.2*\x,1.8) {$1$};}
    \end{scope}
    
    

    \node at (-0.7,0.15) {$\bm{y}$};
    
    \begin{scope}[shift={(0,-1.65)}]
      \draw[dotted, fill=gray, fill opacity=0.3] (-0.2,1.65) rectangle (0,2.0);
      \node at (-0.1,2.2) {$\beta_1$};
      \draw (0,1.65) rectangle (1.2,2.0);
      \node at (0.6,1.8) {$\bm{c}^\prime_1$};
      \draw (1.2,1.65) rectangle (2.0,2.0);
      \foreach \x in {0,1,2,3} {\node at (1.3+0.2*\x,1.8) {$0$};}
      \draw (2.0,1.65) rectangle (2.8,2.0);
      \foreach \x in {0,1,2,3} {\node at (2.1+0.2*\x,1.8) {$1$};}
    \end{scope}
    
    \begin{scope}[shift={(2.8,-1.65)}]
      \draw[dotted, fill=gray, fill opacity=0.3] (-0.4,1.65) rectangle (-0.2,2.0);
      \node at (-0.3,2.2) {$\beta_2$};
      \draw (0,1.65) rectangle (1.0,2.0);
      \node at (0.5,1.8) {$\bm{c}_2$};
      \draw (1.0,1.65) rectangle (1.8,2.0);
      \foreach \x in {0,1,2,3} {\node at (1.1+0.2*\x,1.8) {$0$};}
      \draw (1.8,1.65) rectangle (2.6,2.0);
      \foreach \x in {0,1,2,3} {\node at (1.9+0.2*\x,1.8) {$1$};}
    \end{scope}
    
    \begin{scope}[shift={(5.4,-1.65)}]
      \draw[dotted, fill=gray, fill opacity=0.3] (-0.4,1.65) rectangle (-0.2,2.0);
      \node at (-0.3,2.2) {$\beta_3$};
      \draw (0,1.65) rectangle (0.8,2.0);
      \node at (0.45,1.8) {$\bm{c}^\prime_3$};
      \draw (0.8,1.65) rectangle (1.6,2.0);
      \foreach \x in {0,1,2,3} {\node at (0.9+0.2*\x,1.8) {$0$};}
      \draw (1.6,1.65) rectangle (2.4,2.0);
      \foreach \x in {0,1,2,3} {\node at (1.7+0.2*\x,1.8) {$1$};}
    \end{scope}
    
    \begin{scope}[shift={(7.8,-1.65)}]
      \draw[dotted, fill=gray, fill opacity=0.3] (-0.2,1.65) rectangle (0,2.0);
      \node at (-0.1,2.2) {$\beta_4$};
      \draw (0,1.65) rectangle (1.0,2.0);
      \node at (0.5,1.8) {$\bm{c}^\prime_4$};
      \draw (1.0,1.65) rectangle (1.8,2.0);
      \foreach \x in {0,1,2,3} {\node at (1.1+0.2*\x,1.8) {$0$};}
      \draw (1.8,1.65) rectangle (2.6,2.0);
      \foreach \x in {0,1,2,3} {\node at (1.9+0.2*\x,1.8) {$1$};}
    \end{scope}
    
    \begin{scope}[shift={(10.4,-1.65)}]
      \draw[dotted, fill=gray, fill opacity=0.3] (-0.2,1.65) rectangle (0,2.0);
      \node at (-0.1,2.2) {$\beta_5$};
      \draw (0,1.65) rectangle (1.0,2.0);
      \node at (0.5,1.8) {$\bm{c}_5$};
      \draw (1.0,1.65) rectangle (1.8,2.0);
      \foreach \x in {0,1,2,3} {\node at (1.1+0.2*\x,1.8) {$0$};}
      \draw (1.8,1.65) rectangle (2.6,2.0);
      \foreach \x in {0,1,2,3} {\node at (1.9+0.2*\x,1.8) {$1$};}
    \end{scope}
    
    

    \draw[fill=gray, fill opacity=0.3] (1.4,0) rectangle (2.2,2.0);
    \draw[fill=gray, fill opacity=0.3] (4.0,0) rectangle (4.8,2.0);
    \draw[fill=gray, fill opacity=0.3] (6.6,0) rectangle (7.4,2.0);
    \draw[fill=gray, fill opacity=0.3] (9.2,0) rectangle (10.0,2.0);
    \draw[fill=gray, fill opacity=0.3] (11.8,0) rectangle (12.6,2.0);
    
    \draw[->, >=stealth, thick] (0.6,-0.4) -- (0.6,-0.7);
    \node at (0.6,-1.2) {$\bm{?}^{(t)}$};
    \draw[->, >=stealth, thick] (3.3,-0.4) -- (3.3,-0.7);
    \node at (3.3,-1.2) {$\bm{c}_2$};
    \draw[->, >=stealth, thick] (5.85,-0.4) -- (5.85,-0.7);
    \node at (5.85,-1.2) {$\bm{?}^{(t)}$};
    \draw[->, >=stealth, thick] (8.3,-0.4) -- (8.3,-0.7);
    \node at (8.3,-1.2) {$\bm{c}^\prime_4$};
    \draw[->, >=stealth, thick] (10.9,-0.4) -- (10.9,-0.7);
    \node at (10.9,-1.2) {$\bm{c}_5$};
    
  \end{tikzpicture}
  \caption{Example of Algorithm~\ref{algo:tran_era}}
  \label{fig: ex}
\end{figure*}

\begin{enumerate}[label=\textbf{(b=\arabic*)}, leftmargin=*]
\item Since $l=0$ and $\bm{y}_1 = 0001$,
  Line $10$ is performed:
  $\bm{z}_1 \gets \bm{?}^{(E)}$,
  $P \gets \{ 1 \}$,
  and
  $v \gets 1$.
\item Since $l = 1$ and $\bm{y}_2 = 0001$,
  Line $6$ is performed:
  $\beta_2 \gets E+8$ and
  $\bm{z}_2 \gets \bm{y}_{\setint{ E+10, 2E+9 }} = \bm{c}_2$.
\item Since $l = 1$ and $\bm{y}_3 = 0011$, 
  Line $8$ is performed ($w = 1$):
  $\bm{z}_3 \gets \bm{?}^{(E)}$,
  $P \gets \{ 1, 3 \}$,
  and
  $u \gets 1$.
\item Since $l=0$ and $\bm{y}_4 = 0011$,
  Line $6$ is performed:
  $\beta_4 \gets 3(E+8)$ and
  $\bm{z}_4 \gets \bm{y}_{\setint{ 3E+25, 4E+24 }} = \bm{c}'_4$.
\item Since $l=0$ and $\bm{y}_5 = 0011$,
  Line $6$ is performed:
  $\beta_5 \gets 4(E+8)$ and
  $\bm{z}_5 \gets \bm{y}_{\setint{ 4E+33, 5E+32 }} = \bm{c}_5$.  
\end{enumerate}
Hence,
the output is
$(\bm{?}^{(E)}, \bm{c}_2, \bm{?}^{(E)}, \bm{c}'_4, \bm{c}_5)$ and $P =\{ 1,3 \}$.
In addition,
we get $(\bm{c}_1, \bm{c}_2, \bm{c}_3, \bm{c}_4, \bm{c}_5)$
by the decoder of $C$ for $2$-erasure and $1$-substitution.

\subsection{Decoding Algorithm for Hagiwara Codes}
\label{subsec:deco_hagicode}
\subsubsection{Block-Transformation/Unitary/Erasure Error}
\label{subsubsec:block_error}
In this paper,
we call a transformation to a state in $\state{2}{E}$ from $R_b$
a {\it block-transformation error}.
In particular,
we call
a block-transformation error a {\it block-unitary error}
if
this error is a unitary transformation,
and
a block-unitary error a {\it block-erasure error} at $b$
if
the receiver knows the error position $b$.
For example,
for qubits $q_1,q_2,q_3,q_4$,
a block-transformation (resp.~unitary) error changes
from $R_b = q_1 q_2 q_3 q_4$
to $q_1 q_3 q q_4$
(resp.~$\mat{U} R_b \mat{U}^\dagger$),
where
$q$ is a new qubit
and $\mat{U}$ is unitary.

The decoding of $\QU{R}$ reduces to
the decoding of $B(C_1)$ and $B(C_2)^{\perp}$ ($= B(C_2^\perp)$)
by the decoding algorithm \cite{BB04489264} of the CSS codes;
namely,
that reduces to the decoding of $C_1$ and $C_2^\perp$.
Then,
a block-unitary (resp.~erasure) error correction
is reduced to
a classical substitution (resp.~erasure) error correction
of $C_1$ and $C_2^\perp$.
For the minimum distance $d_1$ of $C_1$ and 
the minimum distance $d_2^\perp$ of $C_2^\perp$,
it is known that
$\QU{R}$ corrects $t = (\min \{ d_1, d_2^\perp \} - 1)$ block-erasures.
Hence,
$\QU{R}$ corrects
$e$ block-erasure errors and $m$ block-unitary errors,
where
$e + 2m \leq t$ \cite{grassl1997codes,gottesman2005quantum}.

\subsubsection{Algorithm}
Let $\sigma$ be
the received state,
$\QU{M}$ be
the measurement of a qubit
by $\{ M_0 = \ket{0}\bra{0}, M_1 = \ket{1}\bra{1} \}$,
and
$\bm{r}_b$ be
a sequence of a measurement outcomes
by $\QU{M}$ of
the $(\gamma_b-\tau_d+1)$, $(\gamma_b-\tau_d+2)$, ..., $(\gamma_b+\tau_i)$th qubits
in $\sigma$.
Note that
$\QU{M}$ does not destroy the qubits in $R_b$
since
the qubits in $R_b$
is not measured.
Then,
by modifying Algorithm~\ref{algo:tran_era},
we get a decoding algorithm for Hagiwara codes by
replacing 
$\bm{y}_b$
with $\bm{r}_b$,
$\bm{y}_{ \setint{ \beta_b+l+1, \beta_b+l+E } }$ in Line $6$
with the $\beta_b+l+1, \beta_b+l+2,..., \beta_b+l+E$th qubits in $\sigma$,
and
$\bm{?}^{(E)}$ in Lines $8,10$
with 
the erasure block.

\subsubsection{Justification}
The output is
a state given by
$e$ block-erasure errors at $P$ and $m$ block-transformation errors
to a codeword in $\QU{R}$,
shown in Section~\ref{subsec:just}.
Note that
any composite error consisting of insertions and deletion
is denoted by
$\QU{I}_P \circ \QU{D}_Q$
for some $P$ and $Q$ \cite{nakamura2026insertion}.
A block-transformation error
is decomposed into the Pauli errors
as follows:
\begin{theo}
  \label{theo:insel_pau_kai}
  Suppose that
  a quantum channel $\QU{I}_Q \circ \QU{D}_P$ changes
  from $\rho \in \state{2}{n}$
  to $\sigma \in \QU{I}_Q(\QU{D}_P(\rho))$,
  where
  $P = \{ p_1, p_2, ..., p_{t_d} \}$ ($p_1 < p_2 < \cdots < p_{t_d}$)
  and
  $Q = \{ q_1, q_2, ..., q_{t_d} \}$ ($q_1 < q_2 < \cdots < q_{t_d}$).
  Furthermore,
  let $\tau$ be a permutation on $\setint{n}$
  such that $\tau(p_i) = q_i$ for $i = 1,2,...,t_d$
  and
  $u < v \Rightarrow \tau(u) < \tau(v)$
  for $u,v \notin P$.
  In addition,
  suppose that
  $S_i \subset \setint{n}$ ($i = 1,2,...,k$, $k \in \Z^+$) is
  a set of continuous positions
  (i.e., $S_i = \setint{u,v}$ for some $u,v$)
  and
  $\tau_i$ is a permutation on $S_i$
  such that
  $S_1, S_2, ..., S_k$ are pairwise disjoint
  and
  $\tau = \tau_1 \circ \tau_2 \circ \cdots \circ \tau_k$.
  Then,
  $\QU{I}_Q \circ \QU{D}_P$ is denoted by
  a complex linear combination of the Pauli errors for the qubits in
  $P \cup \bigcup_{i=1}^k S_i$.
\end{theo}
\begin{coro}
  \label{coro:tran_pau}
  A block-transformation error for $R_b$
  is decomposed into the Pauli errors for qubits in $R_b$.
\end{coro}
These properties are shown in Appendix~\ref{sec:app1}.
From the above,
a decoding algorithm to a block-unitary error
also corrects
a block-transformation error.
Hence,
we can correct
$t_d$-deletions and $t_i$-insertions by Algorithm~\ref{algo:tran_era},
since $\QU{R}$ corrects $e$ block-erasure errors and $m$ block-unitary errors.

\subsection{Justification for Algorithm~\ref{algo:tran_era}}
\label{subsec:just}
This section
proves that
Algorithm~\ref{algo:tran_era}
transforms from $\bm{y}$
to a sequence $\bm{z}$,
which is sequence $(c_1, c_2, ..., c_N) \in B(C)$
with $e$ block-erasures at $P$ and $m$ block-substitutions,
where $e + 2m \leq t$.

\subsubsection{Formularization}
For each $b \in \setint{N}$,
we denote the value of $u$ (resp.~$v$, $l$) in Line $11$
by $u_b$ (resp.~$v_b$, $l_b$),
where
$u_0, v_0 := 0$.
Let $P_e$ (resp.~$P_s$) be 
the positions of the erased (resp.~substituted) blocks in $\bm{y}$,
i.e.,
\begin{align*}
  P_e
  =
  \{
  b \in \setint{N}
  \mid
  &u_{b-1} < u_b
  \lor
  v_{b-1} < v_b
  \},\\
  P_s
  =
  \{
  b \in \setint{N}
  \mid
  &\bm{y}_{ \setint{ \beta_b+l_b+1, \beta_b+l_b+E } }
  \neq
  \bm{c}_b,\\
  &u_{b-1} = u_b,
  v_{b-1} = v_b
  \}.
\end{align*}
Note that
$\bm{y}_{ \setint{ \beta_b+l_b+1, \beta_b+l_b+E } } = \bm{c}_b$
for $b \notin P_e \cup P_s$.
Since $e = |P_e|$ and $m = |P_s|$,
the following theorem
gives the main result of this section.

\begin{theo}
  \label{theo:main}
  $|P_e| + 2|P_s| \leq t$.
\end{theo}

\subsubsection{Notations and their Properties}
Let $J \in \subint{N(E+2t)}{\leq \tau_d}$ (resp.~$K \in \subint{N(E+2t)+r}{\leq \tau_i}$) be
the set of the positions of the deleted (resp.~inserted) symbols
in $\bm{c}$ (resp.~$\bm{y}$)
i.e.,
$\bm{y} \in I_K( D_J(\bm{c}) )$,
and
let $j_b$ be
the number of the deleted symbols in $\bm{c}_{\setint{\gamma_b}}$,
i.e.,
$j_b = |\setint{\gamma_b} \cap J|$.
Then,
the $(\gamma_b-j_b)$th symbol in $D_J(\bm{c})$
is the last $0$ in the $b$th marker
since
$t_d \leq t-1$,
as shown on the left side of Fig.~\ref{fig: special_bits}.
The $(\gamma_b - j_b)$th symbol in $D_J(\bm{c})$ moves
to the $(\gamma_b - j_b + k_b)$th symbol in $\bm{y}$
by the insertion errors,
where $k_b$ is a non-negative integer.
Let $\xi_b$ be
the number of consecutive zeros
following the $(\gamma_b-j_b+\overline{k}_b)$th symbol in $\bm{y} \in I_K(D_J(\bm{c}))$,
as shown on the left side of Fig.~\ref{fig: special_bits}.
Then,
the number of the inserted symbols in $\bm{y}_{\setint{\gamma_b - j_b + \overline{k}_b + \xi_b}}$
equals $(\overline{k}_b + \xi_b) =: k_b$.
To summarize,
$\bm{y}_{\setint{\gamma_b - j_b + k_b}}$ is given
by $j_b$-deletions and $k_b$-insertions to $\bm{c}_{\setint{\gamma_b}}$,
as shown on the left side of Fig.~\ref{fig: special_bits}.

\begin{figure}[t]
  \centering
  \begin{tikzpicture}
    \node at (-0.2,-0.7) {$\bm{c}$};
    \draw[->, thick] (-0.2,-1.1) -- (-0.2,-1.4);
    \node at (-0.2,-1.25)[right] {$D_J$};
    \node at (-0.2,-1.9) {$D_J(\bm{c})$};
    \draw[->, thick] (-0.2,-2.4) -- (-0.2,-2.7);
    \node at (-0.2,-2.55)[right] {$I_K$};
    \node at (-0.2,-3.1) {$\bm{y}$};

    \node at (0.8,-0.7) {$...$};
    \foreach \x in{1,2,3} \draw (1.0+0.2*\x,-0.7) node{$0$};
    \node[fill=gray, fill opacity=0.4, text opacity=1, inner sep=.75pt] at
    (1.0+0.2*4,-0.7){$0$};
    \foreach \x in{1,2,3,4} \draw (1.8+0.2*\x,-0.7) node{$1$};
    \node at (2.9,-0.7) {$...$};
    \draw[<-] (1.8,-0.45) -- (1.8,-0.2) node[above]{\tiny $\gamma_b$};
    
    \node at (0.8,-1.9) {$...$};
    \draw (0.8+0.2*2,-1.9) node{$0$};  
    \node[fill=gray, fill opacity=0.4, text opacity=1, inner sep=.75pt] at
    (0.8+0.2*3,-1.9){$0$};  
    \foreach \x in{1,2,3,4} \draw (1.4+0.2*\x,-1.9) node{$1$};
    \node at (2.5,-1.9) {$...$};
    
    \node at (0.8,-3.1) {$...$};
    \foreach \x in{1,2,3} \draw (1.0+0.2*\x,-3.1) node{$0$};
    \node[fill=gray, fill opacity=0.4, text opacity=1, inner sep=.75pt] at
    (1.0+0.2*4,-3.1){$0$}; 
    \node[draw, inner sep=.75pt] at (1.0+0.2*5,-3.1){$0$};
    \draw (2.2+0.2*0,-3.1) node{$1$};
    \draw (2.2+0.2*1,-3.1) node{$0$};
    \draw (2.2+0.2*2,-3.1) node{$1$};
    \node at (2.9,-3.1) {$...$};

    \draw[dotted] (1.7,-0.6) -- (1.7,-2.0);
    \node at (1.5,-1.2) {\footnotesize $j_b$};
    \draw[->] (1.7,-1.4) -- (1.3,-1.4);
    \draw[dotted] (1.3,-1.2) -- (1.3,-3.2);
    \node at (1.5,-2.3) {\footnotesize $\overline{k}_b$};
    \draw[->] (1.3,-2.5) -- (1.7,-2.5);    
    \draw[dotted] (1.7,-2.4) -- (1.7,-3.2);
    \node at (1.8,-2.7) {\footnotesize $\xi_b$};
    \draw[->] (1.7,-2.9) -- (1.9,-2.9);
    \draw[dotted] (1.9,-2.4) -- (1.9,-3.2);
    \draw[<-] (2,-3.4) -- (2,-3.65)
    node[below]{\tiny $\gamma_b - j_b + k_b$};

    \draw[dotted] (3.5,0.25) -- (3.5,-4);

    \begin{scope}[shift={(4.5,0)}]
      \node at (-0.2,-0.7) {$\bm{c}$};
      \draw[->, thick] (-0.2,-1.1) -- (-0.2,-1.4);
      \node at (-0.2,-1.25)[right] {$D_J$};
      \node at (-0.2,-1.9) {$D_J(\bm{c})$};
      \draw[->, thick] (-0.2,-2.4) -- (-0.2,-2.7);
      \node at (-0.2,-2.55)[right] {$I_K$};
      \node at (-0.2,-3.1) {$\bm{y}$};

      \node at (0.8,-0.7) {$...$};
      \draw (1.0+0.2*1,-0.7) node{$0$};
      \node[fill=gray, fill opacity=0.2, text opacity=1, inner sep=.75pt] at
      (1.0+0.2*2,-0.7) {$0$};
      \node at (1.7,-0.7) {$...$};
      \draw (1.8+0.2*1,-0.7) node{$1$};
      \node[fill=gray, fill opacity=0.4, text opacity=1, inner sep=.75pt] at
      (1.8+0.2*2,-0.7) {$1$};
      \draw (1.8+0.2*3,-0.7) node{$1$};
      \node at (2.7,-0.7) {$...$};
      \node at (1.4,-0.45) {\tiny $\gamma_b$};
      \draw[<-] (2.2,-0.45) -- (2.2,-0.2) node[above]{\tiny $\gamma_b + \tau_i - k'_b + j'_b$};
      
      \node at (0.8,-1.9) {$...$};
      \node[fill=gray, fill opacity=0.2, text opacity=1, inner sep=.75pt] at
      (1.2,-1.9) {$0$};
      \node at (1.5,-1.9) {$...$};
      \node[fill=gray, fill opacity=0.4, text opacity=1, inner sep=.75pt] at
      (1.8,-1.9) {$1$};
      \node at (2,-1.9) {$1$};
      \node at (2.3,-1.9) {$...$};
      
      \node at (0.8,-3.1) {$...$};
      \node[fill=gray, fill opacity=0.2, text opacity=1, inner sep=.75pt] at
      (1.4,-3.1) {$0$};
      \node at (1.7,-3.1) {$...$};
      \node[fill=gray, fill opacity=0.4, text opacity=1, inner sep=.75pt] at
      (2,-3.1) {$1$};
      \node at (2.3,-3.1) {$...$};
      \node[draw, inner sep=.75pt] at (2.6,-3.1){$0$};
      \node at (2.9,-3.1) {$...$};

      \draw[dotted] (2.1,-0.6) -- (2.1,-2.0);
      \node at (1.9,-1.3) {\footnotesize $j'_b$};
      \draw[<-] (1.7,-1.5) -- (2.1,-1.5);
      \draw[dotted] (1.7,-3.2) -- (1.7,-1.3);
      \node at (1.9,-2.3) {\footnotesize $k'_b - \kappa_b$};
      \draw[->] (1.7,-2.5) -- (2.1,-2.5);
      \draw[dotted] (2.1,-2.4) -- (2.1,-3.2);
      \draw[<-] (2.7,-2.9) -- (2.1,-2.9);
      \node at (2.4,-2.7){\footnotesize $\kappa_b$};
      \draw[dotted] (2.7,-2.4) -- (2.7,-3.2);
      \draw[<-] (2.6,-3.4) -- (2.6,-3.65) node[below]{\tiny $\gamma_b + \tau_i$};
    \end{scope}  
  \end{tikzpicture}  
  \vspace{-0.7\baselineskip}
  \caption{Change $\bm{c}_{\setint{\gamma_b}}$ (Left) and $\bm{y}_{\setint{\gamma_b+\tau_i}}$ (Right)}
  \label{fig: special_bits}
\end{figure}

Define $\kappa_b := \min \{ \kappa \ge 0 \mid \gamma_b + \tau_i - \kappa \notin K\}$;
namely,
the $(\gamma_b + \tau_i - \kappa_b + 1), (\gamma_b + \tau_i - \kappa_b + 2), ..., (\gamma_b + \tau_i)$th symbols
in $\bm{y}_{\setint{\gamma_b+\tau_i}}$
are inserted
and
the $(\gamma_b + \tau_i - \kappa_b)$th symbol in $\bm{y}_{\setint{\gamma_b+\tau_i}}$
is not inserted.
Hence,
the number of the inserted symbols in $\bm{y}_{\setint{\gamma_b+\tau_i}}$
equals
$k'_b := (|\setint{\gamma_b+\tau_i-\kappa_b} \cap K| + \kappa_b)$
since
that in $\bm{y}_{\setint{\gamma_b+\tau_i-\kappa_b}}$
equals
$|\setint{\gamma_b+\tau_i-\kappa_b} \cap K|$.
In addition,
the $(\gamma_b + \tau_i - \kappa_b)$th symbol in $\bm{y}_{\setint{\gamma_b+\tau_i}}$
is given by
shifting the $(\gamma_b + \tau_i - k'_b)$th symbol in $D_J(\bm{c})$.
To summarize,
if
for some $j'_b$,
the $(\gamma_b + \tau_i - k'_b)$th symbol in $D_J(\bm{c})$
is given by
shifting the last symbol in $\bm{c}_{\setint{\gamma_b + \tau_i - k'_b + j'_b}}$,
$\bm{y}_{\setint{\gamma_b+\tau_i}}$ is given
by $j'_b$-deletions and $k'_b$-insertions to $\bm{c}_{\setint{\gamma_b + \tau_i - k'_b + j'_b}}$,
as shown on the right side of Fig.~\ref{fig: special_bits}.

In Algorithm~\ref{algo:tran_era},
the process branches based on $\bm{y}_b = \bm{y}_{\setint{\gamma_b-\tau_d+1,\gamma_b+\tau_i}}$.
So,
we define
the following subset $\Omega_b$ (resp.~$\Gamma_b$) of the front (resp.~tail) part of
$\setint{\gamma_b-\tau_d+1,\gamma_b+\tau_i}$.
\begin{align*}
  \Omega_b
  :=
  &\setint{\gamma_b-\tau_d+1,\gamma_b-j_b+k_b-\xi_b-1},\\
  =
  &\setint{\gamma_b-\tau_d+1,\gamma_b-j_b+\overline{k}_b-1},\\
  \Gamma_b
  :=
  &\setint{\gamma_b-j_b+k_b,\gamma_b+\tau_i}
\end{align*}
Recall that the last $0$ of the $b$th marker
is
$(\gamma_b-j_b+\overline{k}_b)$th symbol in $\bm{y}$.
If $y_i = 1$ ($i \in \Omega_b$),
marker-destruction occurs
and
$\bm{y}_b$ cannot be expressed
by the form $\bm{0}^{(\tau_d+l)} \bm{1}^{(\tau_i-l)}$ or $\bm{0}^{(\tau_d+l-w)} \bm{1}^{(\tau_i-l+w)}$.
Hence,
in such case,
Algorithm~\ref{algo:tran_era}
performs
the insertion detection.
Similarly,
if $y_i = 0$ ($i \in \Gamma_b$),
Algorithm~\ref{algo:tran_era}
performs
the insertion detection.
To summarize,
as shown in Table~\ref{tab:algo_for_case},
Algorithm~\ref{algo:tran_era}
performs
the deletion/insertion detection
for each case.

\begin{table}
  \begin{center}
    \caption{Execution of Algorithm~\ref{algo:tran_era} for Each Case}
    \label{tab:algo_for_case}
    \begin{tabular}{l|l|ll}
      \hline
      \multicolumn{2}{c|}{Case} & Execution\\
      \hline
      Marker-Destruction & any & Insertion detection & (Line $10$)\\
      Marker-Preserve & $k_b - j_b = l_{b-1}$ & No detection & (Line $6$)\\
      Marker-Preserve & $k_b - j_b < l_{b-1}$ & Deletion detection & (Line $8$)\\
      Marker-Preserve & $k_b - j_b > l_{b-1}$ & Insertion detection & (Line $10$)\\
      \hline
    \end{tabular}
  \end{center}
\end{table}

\subsubsection{Proof of Theorem~\ref{theo:main}}
We prove the following lemmas
for Theorem~\ref{theo:main}.
We denote
the set of $b \in \setint{N}$ such that
marker-preserve occurs in the $b$th marker
by $M(P)$.

\begin{lem}
  \label{lem:Q_BEDI}
  Define
  \begin{align*}
    Q_B
    &:=
    \{
    b \in M(P)
    \mid
    k_b - j_b = l_{b-1},
    j'_{b-1} = j_b,
    k'_{b-1} = k_b
    \},\\
    Q_E
    &:=
    \{
    b \in M(P)
    \mid
    k_b - j_b = l_{b-1},
    [
      j'_{b-1} < j_b
      \lor
      k'_{b-1} < k_b
    ]
    \},\\
    Q_D
    &:=
    \{
    b \in M(P)
    \mid
    k_b - j_b < l_{b-1}
    \},\\
    Q_I
    &:=
    \{
    b \in \setint{N}
    \mid
    b \notin M(P)
    \lor
    k_b - j_b > l_{b-1}
    \},
  \end{align*}
  where
  $j'_0, k'_0 := 0$.
  Then,
  the following hold:
  \begin{enumerate}
  \item If $b \in Q_B$,
    then
    $y_{\beta_b+l_{b-1}+i} = c_{\beta_b+i}$ for all $i \in \setint{E}$,
    $u_{b-1} = u_b$,
    and
    $v_{b-1} = v_b$
    hold.
  \item If $b \in Q_E$,
    $u_{b-1} = u_b$
    and
    $v_{b-1} = v_b$
    hold.
  \item If $b \in Q_D$,
    $u_{b-1} < u_b$
    holds.
  \item If $b \in Q_I$,
    $v_{b-1} < v_b$
    holds.
  \item $b \in P_e$ iff
    $b \in Q_D \cup Q_I$.
  \item If $b \in P_s$,
    $b \in Q_E$ holds.
  \end{enumerate}
\end{lem}

This lemma is shown in Appendix~\ref{sec:app2}.

\begin{lem}
  \label{lem:nb}
  Define $n_b := | \setint{b} \cap Q_E |$.
  Then,
  \begin{align}
    u_b &\leq j'_b - n_b,
    \label{eq:nb_1}\\
    v_b &\leq k'_b - n_b
    \label{eq:nb_2}
  \end{align}
  hold
  for $b \in \setint{N}$.
\end{lem}
\begin{IEEEproof}
  We prove this
  by induction on $b$.
  We omit the base case $b = 1$
  because
  it is similar to the inductive step shown below.
  Note that
  $\{ Q_B, Q_E, Q_D, Q_I \}$
  is a partition of $\setint{N}$
  and
  \begin{align*}
    0
    &=
    j'_0
    \leq j_1 \leq j'_1 \leq j_2 \leq j'_2
    \leq \cdots \leq 
    j_N \leq j'_N \leq \tau_d,\\
    0
    &=
    k'_0
    \leq k_1 \leq k'_1 \leq k_2 \leq k'_2
    \leq \cdots \leq 
    k_N \leq k'_N \leq \tau_i.
  \end{align*} 
  Assume that
  $u_{b-1} \leq j'_{b-1} - n_{b-1}$, $v_{b-1} \leq k'_{b-1} - n_{b-1}$.
  Then,
  \begin{enumerate}
  \item (Case~:~$b \in Q_B$)
    Since $b \notin Q_E$,
    we get
    $n_b = |\setint{b} \cap Q_E| = |\setint{b-1} \cap Q_E| =n_{b-1}$.
    Since
    $1)$ in Lemma~\ref{lem:Q_BEDI} yields
    $u_{b-1} = u_b$ and $v_{b-1} = v_b$,
    we have $u_b = u_{b-1} \leq j'_{b-1} - n_{b-1} \leq j'_{b} - n_{b}$.
    Similarly,
    we have Eq.~\eqref{eq:nb_2}.
  \item (Case~:~$b \in Q_E$)
    We get 
    $n_b = |\setint{b} \cap Q_E| = |\setint{b-1} \cap Q_E| + 1 = n_{b-1} + 1 > n_{b-1}$
    from $b \in Q_E$
    and
    $u_{b-1} = u_b$ and $v_{b-1} = v_b$
    from $1)$ in Lemma~\ref{lem:Q_BEDI}.
    \begin{enumerate}
    \item (Case~:~$u_{b-1} = j'_{b-1} - n_{b-1}$ and $v_{b-1} = k'_{b-1} - n_{b-1}$)
      Then,
      from the definitions of $l_b$ and $Q_E$,
      \begin{align}
        k'_{b-1} - j'_{b-1} = v_{b-1} - u_{b-1} = l_{b-1} = k_b - j_b
        \label{eq:QE1}
      \end{align}
      holds.
      Here,
      if $j'_{b-1} = j_b$,
      we get $k'_{b-1} = k_b$ from Eq.~\eqref{eq:QE1}
      and
      $k'_{b-1} < k_b$ from the definition of $Q_E$.
      Hence,
      $j'_{b-1} < j_b \leq j'_b$ holds.
      Similarly,
      $k'_{b-1} < k_b$ holds.
      Then,
      we get
      \begin{align*}
        u_b
        =
        u_{b-1}
        =
        j'_{b-1} - n_{b-1}
        <
        j'_b - n_{b-1}
        <
        j'_b - n_b.   
      \end{align*}
    \item (Case~:~$u_{b-1} < j'_{b-1} - n_{b-1}$ and $v_{b-1} = k'_{b-1} - n_{b-1}$)
      Then,
      $k'_{b-1} - j'_{b-1} < v_{b-1} - u_{b-1} = l_{b-1} = k_b - j_b$ holds
      from the definitions of $l_b$ and $Q_E$.
      Since $j'_{b-1} \leq j_b$,
      we get $k'_{b-1} < k_b < k'_b$.
      Hence,
      \begin{align*}
        v_b
        =
        v_{b-1}
        =
        k'_{b-1} - n_{b-1}
        <
        k'_b - n_{b-1}
        <
        k'_b - n_b.
      \end{align*}
    \item (Case~:~$u_{b-1} \leq j'_{b-1} - n_{b-1}$ and $v_{b-1} < k'_{b-1} - n_{b-1}$)
      Similarly to the case $2$-b),
      we get Eq.~\eqref{eq:nb_1}.
    \end{enumerate}
    Similarly,
    we have Eq.~\eqref{eq:nb_2}.
  \item (Case~:~$b \in Q_D$)
    From $b \notin Q_E$,
    $n_b = n_{b-1}$ holds.
    From Line $8$ in Algorithm~\ref{algo:tran_era},
    for $w = u_b - u_{b-1}$,
    $\bm{y}_b = \bm{0}^{(\tau_d+l_b-w)} \bm{1}^{\tau_i - l_b + w}$ holds.
    Hence,
    from $w = (\gamma_b+l_{b-1}) - (\gamma_b+k_b-j_b)$,
    we get
    \begin{align*}
      u_b
      &=
      u_{b-1} + w
      =
      u_{b-1} + l_{b-1} - k_b + j_b
      =
      v_{b-1} - k_b + j_b\\
      &\leq
      k'_{b-1} - n_{b-1} - k_b + j'_b
      \leq
      j'_b - n_b.
    \end{align*}
    In addition,
    we have
    $v_b = v_{b-1} \leq k'_{b-1} - n_{b-1} \leq k'_b - n_b$.
  \item (Case~:~$b \in Q_I$)
    Then,
    $v_b = v_{b-1} +1$ holds.
    From $b \notin Q_E$,
    $n_b = n_{b-1}$ holds.
    If $b \notin M(P)$,
    since insertions occur in the $b$th marker,
    we get $k'_{b-1} < k'_b$
    and
    \begin{align*}
      v_b -1
      =
      v_{b-1}
      \leq
      k'_{b-1} - n_{b-1}
      =
      k'_{b-1} - n_b
      <
      k'_b - n_b,
    \end{align*}
    i.e.,
    $v_b \leq k'_b - n_b$.
    If $b \in M(P)$,
    we get $k_b - j_b > l_{b-1}$
    from Table~\ref{tab:algo_for_case}.
    Hence,
    we have
    \begin{align*}
      k_b - j_b
      >
      l_{b-1}
      =
      v_{b-1} - u_{b-1}
      \ge
      v_{b-1} - j'_{b-1} - n_{b-1}
    \end{align*}
    and
    \begin{align*}
      v_b-1
      &=
      v_{b-1}
      <
      k_b - ( j_b - j'_{b-1}) + n_{b-1}
      \leq
      k_b + n_{b-1}\\
      &=
      k_b + n_b
      \leq
      k'_b + n_b,
    \end{align*}
    i.e.,
    $v_b \leq k'_b + n_b$.
    In addition,
    we get
    $u_b = u_{b-1} \leq j'_{b-1} - n_{b-1} \leq j'_b - n_b$.
  \end{enumerate}
\end{IEEEproof}

\begin{coro}
  \label{coro:<=t}
  Define
  $m_b := |\setint{b} \cap (Q_D \cup Q_I)|$.
  Then,
  $m_b + 2n_b \leq j'_b + k'_b$
  holds
  for $b \in \setint{N}$.
  In particular,
  if $b = N$,
  $|Q_D \cup Q_I| + 2|Q_E| \leq t$ holds.
\end{coro}
\begin{IEEEproof}
  Lemma~\ref{lem:Q_BEDI} yields,
  $u_b + v_b + 2 n_b \leq j'_b + k'_b$.
  From the execution of Algorithm~\ref{algo:tran_era}
  for $b \in Q_D \cup Q_I$,
  we get
  $m_b \leq u_b + v_b$.
  Hence,
  we have
  $m_b + 2 n_b \leq j'_b + k'_b$.
\end{IEEEproof}

\begin{IEEEproof}[Proof of Theorem~\ref{theo:main}]
  We get $|P_e| = |Q_D \cup Q_I|$ from $5)$ in Lemma~\ref{lem:Q_BEDI}
  and
  $|P_s| \leq |Q_E|$ from $6)$ in Lemma~\ref{lem:Q_BEDI}.
  Hence,
  from Corollary~\ref{coro:<=t},
  we have
  $|P_e| + 2|P_s| \leq t$.
\end{IEEEproof}

\section{Conclusion}
This paper proposed a decoding algorithm to
quantum $t_d$-deletion errors and quantum $t_i$-insertion errors
for Hagiwara’s codes.
  
\section*{Acknowledgment}
This work was supported by JSPS KAKENHI Grant Number 22K11905.

\appendices
\section{Proofs of Theorem \ref{theo:insel_pau_kai} and Corollary \ref{coro:tran_pau}}
\label{sec:app1}
A quantum channel $\QU{I}_P \circ \QU{D}_P$ ($P \in \subint{n}{\leq n-1}$)
is not necessarily
a unitary transformation.
For example,
the following transformation is not unitary
since
the rank is not preserved.
\begin{align*}
  \frac{1}{\sqrt{2}}\ket{00}
  +
  \frac{1}{\sqrt{2}}\ket{11}
  &\xrightarrow{\text{$2$nd qubit deletion}}
  \frac{1}{2}\ket{0}\bra{0}
  +
  \frac{1}{2}\ket{1}\bra{1}\\
  &\xrightarrow{\text{$2$nd qubit insertion}}
  \frac{1}{2}\ket{00}\bra{00}
  +
  \frac{1}{2}\ket{10}\bra{10}
\end{align*}
Although 
$\QU{I}_P \circ \QU{D}_P$ is not necessarily unitary,
$\QU{I}_P \circ \QU{D}_P$ is denoted by the Pauli errors
as follows:
\begin{theo}
  \label{theo:insel_pau}
  Suppose that
  a quantum channel $\QU{I}_P \circ \QU{D}_P$ ($P \in \subint{n}{\leq n-1}$) changes
  from $\rho \in \state{2}{n}$
  to $\sigma \in \QU{I}_P(\QU{D}_P(\rho))$.
  Then,
  this error is denoted by
  a complex linear combination of the Pauli errors for the qubits in $P$.  
\end{theo}
\begin{IEEEproof}
  We prove for $P = \{ n \}$.
  If $\QU{D}_{\{ n \}}(\rho)$ is spectrally decomposed as
  $\sum_{ i = 0 }^k \lambda_i \ket{\lambda_i}\bra{\lambda_i}$ ($k := 2^{n-1}$),
  $\rho$ and $\sigma$ are denoted by
  $\rho = \sum_{ i,j = 0 }^k \sqrt{\lambda_i\lambda_j} \ket{\lambda_i}\bra{\lambda_j} \otimes \mat{A}_{i,j}$
  and
  $\sigma = \sum_{ i,j = 0 }^k \sqrt{\lambda_i\lambda_j} \ket{\lambda_i}\bra{\lambda_j} \otimes \mat{B}_{i,j}$ \cite{nakamura2026insertion}.
  Then,
  by the state $\rho_{\text{env}}$ of the environment
  and
  some unitary matrix $\mat{U}$,
  $\sigma$ is denoted as follows:
  \begin{align*}
    \sigma
    &=
    \text{tr}_{\text{env}}
    \left(
    (
    \mat{I}^{(n-1)} \otimes \mat{U}
    )
    (
    \rho \otimes \rho_{\text{env}}
    )
    (
    \mat{I}^{(n-1)} \otimes \mat{U}^\dagger
    )
    \right)\\
    &=
    \textstyle
    \sum_{ i,j = 0 }^k
    \sqrt{ \lambda_i \lambda_j } 
    \ket{\lambda_i}\bra{\lambda_j}
    \otimes
    \text{tr}_{\text{env}}
    (
    \mat{U}
    ( \mat{A}_{i,j} \otimes \rho_{\text{env}} )
    \mat{U}^\dagger
    );
  \end{align*}
  i.e.,
  $\mat{B}_{i,j} = \text{tr}_{\text{env}}( \mat{U} ( \mat{A}_{i,j} \otimes \rho_{\text{env}} ) \mat{U}^\dagger ) =: f(\mat{A}_{i,j})$.
  Since $f$ is linear,
  $f(\mat{A}_{i,j})$ is denoted by
  $f(\mat{A}_{i,j}) = \sum_{a = 1}^u \mat{M}_{1,a} \mat{A}_{i,j} \mat{M}_{2,a}$
  for some $2 \times 2$ matrices $\mat{M}_{1,a}, \mat{M}_{2,a}$
  and some integer $u$
  \cite{watrous2018theory}.
  Hence,
  since $\mat{M}_{1,a}, \mat{M}_{2,a}$ are denoted by
  a complex linear combination of the Pauli errors,
  $\QU{I}_{\{n\}} \circ \QU{D}_{\{n\}}$ is denoted by
  a complex linear combination of the Pauli errors for the $n$th qubit.
  The proof is similar to the above  
  if $P$ is a tail part of $\setint{n}$,
  i.e.
  $P = \setint{i,n}$
  for some $i \in \setint{n}$.
  The proof for any $P$ is given by
  swapping of qubits
  for the results above.
\end{IEEEproof}

$\QU{I}_Q \circ \QU{D}_P$ is denoted by
$\QU{I}_P \circ \QU{D}_P$ and
the permutation $\tau$ in Theorem \ref{theo:insel_pau_kai}.
We give an example.
  
\begin{exam}
  \label{exam:change_error}
  We denote $\rho \in \state{2}{10}$
  by the sequence of qubits $q_1, q_2, ..., q_{10}$.
  Suppose that
  $P = \{ 2,4,6,10 \}$, $Q = \{ 1,6,7,10 \}$,
  and
  $\rho \in \state{2}{10}$ is changed to $\sigma \in \QU{I}_Q (\QU{D}_P(\rho))$
  as shown Fig.~\ref{fig:tikz_cha_R_S},
  where
  we denote the inserted qubit by \textcolor{red}{$q$}.  
  \begin{figure}[t]
    \centering
    \begin{tikzpicture}
      \node at (0,0.75) [left]{$\rho:$};
      \foreach \x in {2,4,6,10} { \node at (0.5*\x,0.5) [above,blue]{$q_{\x}$}; }
      \foreach \x in {1,3,5,7,8,9} { \node at (0.5*\x,0.5) [above]{$q_{\x}$}; }
      \node at (3,0.25) [left]{\tiny $\QU{D}_P$};
      \draw [-latex](3,0.5)--(3,0);
      \foreach \x/\n in {1/1,2/3,3/5,4/7,5/8,6/9}{ \node at (0.5*\x,0)[below]{$q_{\n}$}; }
      \node at (3,-0.75) [left]{\tiny $\QU{I}_Q$};
      \draw [-latex](3,-0.5)--(3,-1);
      \node at (0,-1.25) [left]{$\sigma:$};
      \foreach \x in {1,6,7,10} { \node at (0.5*\x,-1) [below,red]{$q$}; }
      \foreach \x/\n in {2/1,3/3,4/5,5/7,8/8,9/9}{ \node at (0.5*\x,-1)[below]{$q_{\n}$};}
    \end{tikzpicture}
    \caption{Example of a Change by $\QU{I}_Q \circ \QU{D}_P$}
    \label{fig:tikz_cha_R_S}
  \end{figure}
  For the permutation 
  \begin{align*}
    \tau
    &=
    \begin{pmatrix}
      1 & 2 & 3 & 4 & 5 & 6 & 7 & 8 & 9 & 10\\
      2 & 1 & 3 & 5 & 7 & 4 & 6 & 8 & 9 & 10
    \end{pmatrix},    
  \end{align*}
  $\sigma \in \QU{I}_Q (\QU{D}_P(\rho))$
  is a changed state from some $\pi \in \QU{I}_P (\QU{D}_P(\rho))$
  by
  the shifting of $i$th qubit of $\pi$ to $\tau(i)$th,
  as shown Fig.~\ref{fig:tikz_dis_denote}.
  \begin{figure}[t]
    \centering
    \begin{tikzpicture}
      \node at (0,0.75) [left]{$\rho:$};
      \foreach \x in {2,4,6,10} { \node at (0.5*\x,0.5) [above,blue]{$q_{\x}$}; }
      \foreach \x in {1,3,5,7,8,9} { \node at (0.5*\x,0.5) [above]{$q_{\x}$}; }
      \node at (3,0.25) [left]{\tiny $\QU{I}_P \circ \QU{D}_P$};
      \draw [-latex](3,0.5)--(3,0);
      \foreach \x in {2,4,6,10} { \node at (0.5*\x,0) [below,red]{$q$}; }
      \foreach \x in {1,3,5,7,8,9} { \node at (0.5*\x,0) [below]{$q_{\x}$}; }
      \node at (3,-0.75) [left]{\tiny Shifting by $\tau$};
      \draw [-latex](3,-0.5)--(3,-1);
      \node at (0,-1.25) [left]{$\sigma:$};
      \foreach \x in {1,6,7,10} { \node at (0.5*\x,-1) [below,red]{$q$}; }
      \foreach \x/\n in {2/1,3/3,4/5,5/7,8/8,9/9}{ \node at (0.5*\x,-1)[below]{$q_{\n}$};}
    \end{tikzpicture}
    \caption{Example of a Change by $\QU{I}_P \circ \QU{D}_P$ and $\tau$}
    \label{fig:tikz_dis_denote}
  \end{figure}
  By the way,
  given $\rho$ and $\sigma$,
  such $\pi$ is given by
  the shifting of $i$th qubit of $\sigma$ to $\tau^{-1}(i)$th.
\end{exam}

A permutation $\tau$ on $\setint{n}$
is denoted as
\begin{align*}
  \tau = \tau_1 \circ \tau_2 \circ \cdots \tau_k
\end{align*}
by
some subsets $S_i \subset \setint{n}$ and 
some permutations $\tau_i$ on $S_i$.
For example,
$\tau$ in Example~\ref{exam:change_error}
satisfies
\begin{align*}
  \tau
  &=
  \begin{pmatrix}
    1 & 2 & 3 & 4 & 5 & 6 & 7 & 8 & 9 & 10\\
    2 & 1 & 3 & 5 & 7 & 4 & 6 & 8 & 9 & 10
  \end{pmatrix}\\
  &=
  \begin{pmatrix}
    1 & 2\\
    2 & 1
  \end{pmatrix}
  \circ
  \begin{pmatrix}
    4 & 5 & 6 & 7\\
    5 & 7 & 4 & 6
  \end{pmatrix},
\end{align*}
i.e.,
$\tau$ is denoted by $\tau_1 = \begin{pmatrix} 1 & 2\\ 2 & 1 \end{pmatrix}$ on $S_1 = \setint{1,2}$ and $\tau_2 = \begin{pmatrix} 4 & 5 & 6 & 7\\ 5 & 7 & 4 & 6 \end{pmatrix}$ on $S_2 = \setint{4,7}$.
In general,
$\tau$ is decomposed into $\tau_i$ on $S_i$ made in Algorithm~\ref{algo:make_S_i}.
We describe the principle of Algorithm~\ref{algo:make_S_i}.
A permutation $\tau$ on $\setint{n}$ is denoted by
some permutations on $T_j := \setint{p_j,q_j} \cup \setint{q_j,p_j}$
($j \in \setint{t}$, $p_j \neq q_j$).
Hence,
Algorithm~\ref{algo:make_S_i} makes $S_1, S_2, ..., S_k$
by merging of the overlapping sets $T_j, T_{j+1}$
in Lines $10$--$14$.
\begin{figure}[t]
  \begin{algorithm}[H]
    \caption{Making $S_i$}
    \label{algo:make_S_i}
    \begin{algorithmic}[1] 
      \REQUIRE $t \in \PZ$, $P = \{ p_1, p_2, ..., p_{t} \}$, 
      $Q = \{ q_1, q_2, ..., q_{t} \}$
      \ENSURE $k \in \{ 0 \} \cup \PZ$, $\{ S_i \}_{i=1}^k$
      \FOR{$j = 1,2,..., t$}
      \IF{$p_j \neq q_j$}
      \STATE $T_j \gets \setint{p_j,q_j} \cup \setint{q_j,p_j}$
      \ELSE
      \STATE $T_j \gets \emptyset$
      \ENDIF
      \ENDFOR
      \STATE $k \gets 0$
      \FOR{$j = 1,2,..., t-1$}
      \IF{$T_j \cap T_{j+1} \neq \emptyset$}
      \STATE $T_{j+1} \gets T_j \cup T_{j+1}$, $T_j \gets \emptyset$
      \ELSIF{$T_j \neq \emptyset$}
      \STATE $k \gets k+1$, $S_k \gets T_j$
      \ENDIF
      \ENDFOR
      \IF{$T_t \neq \emptyset$}
      \STATE $k \gets k+1$, $S_k \gets T_t$
      \ENDIF
    \end{algorithmic}
  \end{algorithm}
\end{figure}
For example,
for $P, Q$ in Example~\ref{exam:change_error},
Algorithm~\ref{algo:make_S_i} performs as follows:
\begin{enumerate}
\item (Lines $1$--$7$)
  Makes the following sets $T_j$.
  \begin{center}
    \begin{tabular}{c|cccc}
      $j$ & $1$ & $2$ & $3$ & $4$\\
      \hline
      $T_j$ & $\setint{1,2}$  & $\setint{4,6}$  & $\setint{6,7}$  & $\emptyset$ 
    \end{tabular}
  \end{center}
\item (Lines $8$--$15$)
  \begin{enumerate}[label=\textbf{(j=\arabic*)}]
  \item Since $T_1 \cap T_2 = \emptyset$ and $T_1 \neq \emptyset$,
    performs $S_1 \gets T_1 = \setint{1,2}$.
  \item Since $T_2 \cap T_3 \neq \emptyset$,
    performs $T_2 \gets \emptyset$, $T_3 \gets T_2 \cup T_3 = \setint{4,7}$.
  \item Since $T_3 \cap T_4 = \emptyset$ and $T_3 \neq \emptyset$,
    performs $S_2 \gets T_3 = \setint{4,7}$.
  \end{enumerate}
\item (Lines $16$--$18$)
  Since $T_4 = \emptyset$,
  performs nothing.
\end{enumerate}
To summarize,
these outputs are $k = 2$ and
\begin{align*}
  &S_1
  =
  \setint{1,2},
  &&S_2
  =
  \setint{4,7}.
\end{align*}

As shown the following,
$\tau$ is denoted by
a complex linear combination of the Pauli errors for the qubits in $\bigcup_{i=1}^k S_i$
since $\tau_i$ on $S_i$ is denoted by
a unitary transformation for the qubits in $S_i$.
Hence,
$\QU{I}_Q \circ \QU{D}_P$ is denoted by
a complex linear combination of the Pauli errors for the qubits in
$P \cup \bigcup_{i=1}^k S_i$,
since
$\QU{I}_Q \circ \QU{D}_P$ is denoted by
$\QU{I}_P \circ \QU{D}_P$ and $\tau$,
as follows:

\begin{IEEEproof}[Proof of Theorem \ref{theo:insel_pau_kai}]
  For some $\pi \in \QU{I}_P( \QU{D}_P(\rho) )$,
  we get $\sigma \in \QU{I}_Q( \QU{D}_P(\rho) )$
  by
  shifting of the $i$th qubit of $\pi$ to $\tau(i)$th,
  i.e.,
  $\QU{I}_Q \circ \QU{D}_P$ is denoted by
  the composition of 
  $\QU{I}_P \circ \QU{D}_P$ and this shifting.
  This shifting is denoted by 
  \begin{align*}
    \mat{U}_{\tau}
    =
    \sum_{\bm{x} \in \{0,1\}^n}
    \ket{ x_{\tau(1)} x_{\tau(2)} ... x_{\tau(n)} }
    \bra{\bm{x}},
  \end{align*}
  i.e.,
  $\sigma = \mat{U}_{\tau} \pi \mat{U}_{\tau}^\dagger$.
  If $\tau = \tau_{\text{head}} \circ \tau_{\text{tail}}$ holds,
  where
  $\tau_{\text{head}}$ is a permutation on $\setint{i}$
  and
  $\tau_{\text{tail}}$ is a permutation on $\setint{i+1,n}$
  for some $i \in \setint{n}$,
  we get
  \begin{align*}
    \mat{U}_{\tau}
    =
    &\left(
    \sum_{\bm{x}_{\setint{i}} \in \{0,1\}^i}
    \ket{ x_{\tau(1)} x_{\tau(2)} ... x_{\tau(i)} }
    \bra{\bm{x}_{\setint{i}}}
    \right) \otimes \\
    &\left(
    \sum_{\bm{x}_{\setint{i+1,n}} \in \{0,1\}^{n-i}}
    \ket{ x_{\tau(i+1)} x_{\tau(i+2)} ... x_{\tau(n)} }
    \bra{\bm{x}_{\setint{i+1,n}}}
    \right).
  \end{align*}
  In particular,
  if $\tau(i) = i$ holds,
  i.e.,
  $\tau_{\text{head}}$ is a permutation on $\setint{i-1}$
  and
  $\tau_{\text{tail}}$ is a permutation on $\setint{i+1,n}$,
  we get
  \begin{align*}
    \mat{U}_{\tau}
    =
    &\left(
    \sum_{\bm{x}_{\setint{i-1}} \in \{0,1\}^{i-1}}
    \ket{ x_{\tau(1)} x_{\tau(2)} ... x_{\tau(i-1)} }
    \bra{\bm{x}_{\setint{i-1}}}
    \right)\\
    & \otimes 
    \left(
    \sum_{x_i \in \{0,1\}}
    \ket{ x_{\tau(i)} }
    \bra{x_i}
    \right) \otimes \\
    &\left(
    \sum_{\bm{x}_{\setint{i+1,n}} \in \{0,1\}^{n-i}}
    \ket{ x_{\tau(i+1)} x_{\tau(i+2)} ... x_{\tau(n)} }
    \bra{\bm{x}_{\setint{i+1,n}}}
    \right)\\
    =
    &\left(
    \sum_{\bm{x}_{\setint{i-1}} \in \{0,1\}^{i-1}}
    \ket{ x_{\tau(1)} x_{\tau(2)} ... x_{\tau(i-1)} }
    \bra{\bm{x}_{\setint{i-1}}}
    \right) \otimes \mat{I} \otimes \\
    &\left(
    \sum_{\bm{x}_{\setint{i+1,n}} \in \{0,1\}^{n-i}}
    \ket{ x_{\tau(i+1)} x_{\tau(i+2)} ... x_{\tau(n)} }
    \bra{\bm{x}_{\setint{i+1,n}}}
    \right),
  \end{align*}
  where $\mat{I}$ is the $2 \times 2$ identity matrix.
  Similarly,
  since $\tau = \tau_1 \circ \tau_2 \circ \cdots \circ \tau_k$,
  $\mat{U}_{\tau}$ is denoted by
  a complex linear combination of the Pauli errors for the qubits in $\bigcup_{i=1}^k S_i$.
  From Theorem~\ref{theo:insel_pau},
  $\QU{I}_Q \circ \QU{D}_P$ is denoted by
  a complex linear combination of the Pauli errors for the qubits in
  $P \cup \bigcup_{i=1}^k S_i$
  since 
  $\QU{I}_Q \circ \QU{D}_P$ equals
  the composition of 
  $\QU{I}_P \circ \QU{D}_P$ and $\mat{U}_{\tau}$.
\end{IEEEproof}

\begin{IEEEproof}[Proof of Corollary \ref{coro:tran_pau}]
  In a block-transformation error for $R_b$,
  $P$ is a subset of range of $R_b$
  and
  $\tau$ is a permutation on the range.
  Hence,
  from Theorem \ref{theo:insel_pau_kai},
  a block-transformation error for $R_b$
  is denoted by
  a complex linear combination of the Pauli errors for the qubits in $R_b$.
\end{IEEEproof}

\section{Proof of Lemma \ref{lem:Q_BEDI}}
\label{sec:app2}
\begin{enumerate}
\item 
  Recall that
  $\bm{y}_{\setint{\gamma_b - j_b + k_b}}$ is given
  by $j_b$-deletions and $k_b$-insertions to $\bm{c}_{\setint{\gamma_b}}$
  and
  $\bm{y}_{\setint{\gamma_b+\tau_i}}$ is given
  by $j'_b$-deletions and $k'_b$-insertions to $\bm{c}_{\setint{\gamma_b + \tau_i - k'_b + j'_b}}$.
  From $j_{b-1}' = j_b$ and $k'_{b-1} = k_b$,
  no insertions or deletions have occurred in $\bm{c}_b$ and the $b$th marker.
  Hence,
  for all $i \in \setint{E}$,
  the $(\beta_b+i)$th symbol in $\bm{c}$
  is shifted
  to the $(\beta_b+k'_{b-1}-j_{b-1}'+i)$th symbol in $\bm{y}$.
  We get $k'_{b-1}-j_{b-1}' = l_{b-1}$
  since $b \in Q_B$.
  From the above,
  $y_{\beta_b+l_{b-1}+i} = c_{\beta_b+i}$ holds for all $i \in \setint{E}$.
  
  Since 
  marker-preserve occurs in the $b$th marker,
  the $b$th marker $\bm{c}_{\setint{\gamma_b-t+1,\gamma_b+t}}$
  is shifted
  to
  $\bm{y}_{\setint{\gamma_b-t+1-l_{b-1},\gamma_b+t-l_{b-1}}}$.
  Hence,
  the symbols in $\bm{y}_{\setint{\gamma_b-t+1-l_{b-1},\gamma_b-l_{b-1}}}$
  are zeros
  and
  the symbols in $\bm{y}_{\setint{\gamma_b-l_{b-1}+1,\gamma_b+t-l_{b-1}}}$
  are ones,
  i.e.,
  \begin{align*}
    \bm{y}_{\setint{\gamma_b-\tau_d+1,\gamma_b+\tau_i}}
    &=
    \bm{0}^{(\tau_d+l_{b-1})} \bm{1}^{(\tau_i-l_{b-1})}
    =
    \bm{0}^{(\tau_d+l_b)} \bm{1}^{(\tau_i-l_b)}.
  \end{align*}
  Therefore,
  Algorithm~\ref{algo:tran_era}
  does not change the values of $u$ and $v$,
  i.e.,
  $u_b = u_{b-1}$
  and
  $v_b = v_{b-1}$.
\item
  Similarly to $(1)$.
\item
  From $k_b - j_b < l_{b-1}$,
  deletions have occurred in $\bm{c}_b$ and the $b$th marker.    
  Since 
  marker-preserve occurs in the $b$th marker,
  the format of the $b$th marker is preserved
  and
  the $b$th marker is shifted to the left.
  Hence,
  $\bm{y}_{\setint{\gamma_b-\tau_d+1,\gamma_b+\tau_i}} = \bm{0}^{(\tau_d+l-w)} \bm{1}^{(\tau_i-l+w)}$
  holds
  for some $w \in \PZ$
  and
  Algorithm~\ref{algo:tran_era} changes the value of $u$,
  i.e.,
  $u_{b-1} < u_b$.
\item
  If $b \in M(P)$,
  this is shown similarly to $(3)$.
  If $b \notin M(P)$,
  the format of the $b$th marker is not preserved.
  In other words,
  $\bm{y}_{\setint{\gamma_b-\tau_d+1,\gamma_b+\tau_i}} = \bm{0}^{(\tau_d+l_b)} \bm{1}^{(\tau_i-l_b)}, \bm{0}^{(\tau_d+l-w)} \bm{1}^{(\tau_i-l+w)}$
  does not hold.
  Hence,
  Algorithm~\ref{algo:tran_era} changes the value of $v$,
  i.e.,
  $v_{b-1} < v_b$.
\item
  Recall that
  $\{ Q_B, Q_E, Q_D, Q_I \}$
  is a partition of $\setint{N}$.
  From $1)$--$4)$ in Lemma \ref{lem:Q_BEDI},
  we get
  \begin{align*}
    b \in Q_D \cup Q_I
    &\iff [u_{b-1} < u_b] \lor [v_{b-1} < v_b]\\
    &\iff b \in P_e.
  \end{align*}
\item
  From $1)$--$4)$ in Lemma \ref{lem:Q_BEDI},
  $P_s \subset Q_B \cup Q_E$ holds
  since $b \in P_s$
  satisfies
  $u_{b-1} = u_b$ and $v_{b-1} = v_b$.
  If $b \in Q_B$,
  $y_{\beta_b+l_{b-1}+i} = y_{\beta_b+l_b+i} = c_{\beta_b+i}$ holds ($i \in \setint{E}$)
  from $1)$ in Lemma \ref{lem:Q_BEDI},
  i.e.,
  $\bm{y}_{ \setint{ \beta_b+l_b+1, \beta_b+l_b+E } } = \bm{c}_b$.
  Since
  $b \in P_s$
  satisfies
  $\bm{y}_{ \setint{ \beta_b+l_b+1, \beta_b+l_b+E } } \neq \bm{c}_b$,
  $P_s \cap Q_B = \emptyset$ holds,
  i.e.,
  $P_s \subset Q_E$.
  \hfill$\blacksquare$
\end{enumerate}

\bibliographystyle{IEEEtran}
\bibliography{IEEEabrv,myref}

@article{hagiwara2023quantum,
  title={Quantum Deletion Codes Derived From Quantum Reed-Solomon Codes},
  author={Hagiwara, Manabu},
  journal={arXiv preprint arXiv:2306.13399},
  year={2023}
}

@article{hagiwara2025quantum,
  title={Quantum Multi Deletion Codes Derived from Quantum Reed-Solomon Codes},
  author={Hagiwara, Manabu},
  journal={2025 IEEE International Symposium on Information Theory (ISIT)},
  pages={1--6},
  year={2025},
  organization={IEEE}
}

@article{grassl1999quantum,
  title={Quantum Reed-Solomon codes},
  author={Grassl, Markus and Geiselmann, Willi and Beth, Thomas},
  journal={International Symposium on Applied Algebra, Algebraic Algorithms, and Error-Correcting Codes},
  pages={231--244},
  year={1999},
  organization={Springer}
}

@article{levenshtein1966binary,
  author  = {Levenshtein, Vladimir I.},
  title   = {Binary codes capable of correcting deletions, insertions, and reversals},
  journal = {Soviet Physics Doklady},
  volume  = {10},
  pages   = {707--710},
  year    = {1966}
}

@article{calderbank1996good,
  title={Good quantum error-correcting codes exist},
  author={Calderbank, A Robert and Shor, Peter W},
  journal={Physical Review A},
  volume={54},
  number={2},
  pages={1098},
  year={1996},
  publisher={APS}
}

@article{steane1996multiple,
  title={Multiple-particle interference and quantum error correction},
  author={Steane, Andrew},
  journal={Proceedings of the Royal Society of London. Series A: Mathematical, Physical and Engineering Sciences},
  volume={452},
  number={1954},
  pages={2551--2577},
  year={1996},
  publisher={The Royal Society London}
}

@article{leahy2019quantum,
  title={Quantum insertion-deletion channels},
  author={Leahy, Janet and Touchette, Dave and Yao, Penghui},
  journal={arXiv preprint arXiv:1901.00984},
  year={2019}
}

@article{nakayama2020first,
  title={The first quantum error-correcting code for single deletion errors},
  author={Nakayama, Ayumu and Hagiwara, Manabu},
  journal={IEICE Communications Express},
  volume={9},
  number={4},
  pages={100--104},
  year={2020},
  publisher={The Institute of Electronics, Information and Communication Engineers}
}

@article{grassl1997codes,
  title={Codes for the quantum erasure channel},
  author={Grassl, Markus and Beth, Th and Pellizzari, Thomas},
  journal={Physical Review A},
  volume={56},
  number={1},
  pages={33},
  year={1997},
  publisher={APS}
}

@article{it2025-3,
	title     = {Single-Insertion Plus Single-Deletion Correcting Algorithm for Quantum Deletion Correcting Codes Based on Quantum Reed-Solomon Codes},
	author    = {Koki Sasaki and Takayuki Nozaki},
	journal = {IEICE Tech. Rep.},
	volume    = {125},
	number    = {37},
	pages     = {13-18},
	year      = {2025}
}

@article{nakamura2026insertion,
  title={Insertion Correcting Capability for Quantum Deletion-Correcting Codes},
  author={Nakamura, Ken and Nozaki, Takayuki},
  journal={arXiv preprint arXiv:2602.20635},
  year={2026}
}

@book{watrous2018theory,
  title={The Theory of Quantum Information},
  author={Watrous, John},
  year={2018},
  publisher={Cambridge university press}
}

@article{nakamura2025multiple,
  title={Multiple-Insertion-Correcting Non-Binary Quantum Codes and Decoding Algorithm},
  author={Nakamura, Ken and Nozaki, Takayuki},
  journal={IEICE Transactions on Fundamentals of Electronics, Communications and Computer Sciences},
  volume={108},
  number={2},
  pages={123--128},
  year={2025},
  publisher={The Institute of Electronics, Information and Communication Engineers}
}

@article{nakamura2024decoding,
  title={Decoding Algorithm Correcting Single-Insertion Plus Single-Deletion for Non-binary Quantum Codes},
  author={Nakamura, Ken and Nozaki, Takayuki},
  journal={2024 International Symposium on Information Theory and Its Applications (ISITA)},
  pages={86--91},
  year={2024},
  organization={IEEE}
}

@article{gottesman2005quantum,
  title={Quantum error correction and fault-tolerance},
  author={Gottesman, Daniel},
  journal={Quantum Information Processing: From Theory to Experiment},
  volume={199},
  pages={159},
  year={2006},
  publisher={IOS Press}
}

@article{matsumoto2022constructions,
  title={Constructions of $l$-adic $t$-deletion-correcting quantum codes},
  author={Matsumoto, Ryutaroh and Hagiwara, Manabu},
  journal={IEICE Transactions on Fundamentals of Electronics, Communications and Computer Sciences},
  volume={105},
  number={3},
  pages={571--575},
  year={2022},
  publisher={The Institute of Electronics, Information and Communication Engineers}
}

@article{DBLP:conf/isita/SasakiN24,
  author       = {Koki Sasaki and
                  Takayuki Nozaki},
title = {Insertion Correcting Algorithm for Quantum Deletion Correcting Codes
                  Based on Quantum Reed-Solomon Codes},
  journal    = {2024 International Symposium on Information Theory and Its Applications (ISITA)},
  pages        = {92--97},
  year         = {2024}
}

@book{BB04489264,
  author    = "Nielsen, Michael A. and Chuang, Isaac L.",
  title     = "Quantum computation and quantum information",
  publisher = "Cambridge University Press",
  year      = "2010",
}

@article{hagiwara2020four,
  title={A four-qubits code that is a quantum deletion error-correcting code with the optimal length},
  author={Hagiwara, Manabu and Nakayama, Ayumu},
  journal={2020 IEEE International Symposium on Information Theory (ISIT)},
  pages={1870--1874},
  year={2020},
  organization={IEEE}
}

@article{hagiwara2021four,
  title={The four qubits deletion code is the first quantum insertion code},
  author={Hagiwara, Manabu},
  journal={IEICE Communications Express},
  volume={10},
  number={5},
  pages={243--247},
  year={2021},
  publisher={The Institute of Electronics, Information and Communication Engineers}
}
\end{document}